\documentclass[11pt]{article}
\usepackage{amsmath,amsfonts,amsthm,url}
\usepackage[usenames,dvipsnames]{color}
\usepackage{enumerate}
\usepackage{algorithm}
\usepackage[noend]{algorithmic}
\usepackage{fullpage}
\newtheorem{theorem}{Theorem}
\newtheorem{conjecture}{Conjecture}
\newtheorem{corollary}{Corollary}

\newtheorem{lemma}{Lemma}
\newtheorem{proposition}{Proposition}

\theoremstyle{definition}
\newtheorem{definition}{Definition}

\theoremstyle{plain}

\newcommand{\calF}{\mathcal{F}}
\newcommand{\calH}{\mathcal{H}}
\newcommand{\calI}{\mathcal{I}}
\newcommand{\calJ}{\mathcal{J}}
\newcommand{\calP}{\mathcal{P}}
\newcommand{\calS}{\mathcal{S}}

\newcommand{\calB}{\mathcal{B}}

\newcommand{\calD}{\mathcal{D}}

\newcommand{\calDIW}{\mathcal{D}_{\mathcal{I}}^{W}}
\newcommand{\dtv}{\mathrm{d_{TV}}}
\newcommand{\E}{\mathop{\mathbf{E}}}
\newcommand{\eps}{\epsilon}
\newcommand{\Inf}{\mathrm{Inf}}
\newcommand{\ignore}[1]{}
\newcommand{\poly}{\mathrm{poly}}

\newcommand{\SymInf}{\mathrm{SymInf}}
\newcommand{\Var}{\mathop{\mathbf{Var}}}
\newcommand{\core}{\mathrm{core}}
\newcommand{\dist}{\mathrm{dist}}
\newcommand{\layer}[3]{L^{#1}_{\overline{#2}\leftarrow #3}}

\title{Partially Symmetric Functions are Efficiently Isomorphism-Testable}

\author{Eric Blais\footnote{
      School of Computer Science, 
      Carnegie Mellon University,
      Pittsburgh, PA, USA.
      Email: \texttt{eblais@cs.cmu.edu}}
  \and
  Amit Weinstein\footnote{
      Blavatnik School of Computer Science,
      Tel Aviv University,
      Tel Aviv, Israel.
      Email: \texttt{amitw@tau.ac.il}.
      Research supported in part by an ERC Advanced grant.}
  \and 
  Yuichi Yoshida\footnote{
      School of Informatics, Kyoto University and
      Preferred Infrastructure, Inc.,
      Kyoto, Japan.
      Email: \texttt{yyoshida@kuis.kyoto-u.ac.jp}}
}

\begin{document}
\maketitle

\begin{abstract}
Given a function $f : \{0,1\}^n \to \{0,1\}$, the $f$-\emph{isomorphism testing} problem requires a
randomized algorithm to distinguish functions that are identical to $f$ up to relabeling
of the input variables from functions that are far from being so. An important open question in property
testing is to determine for which functions $f$ we can test $f$-isomorphism with a constant 
number of queries. Despite much recent attention to this question, essentially only two classes of
functions were known to be efficiently isomorphism testable: symmetric functions and juntas.  

We unify and extend these results by showing that all \emph{partially symmetric} functions---functions invariant to the reordering of all but a constant number of their variables---are efficiently isomorphism-testable. This class of functions, first introduced by Shannon, includes symmetric functions, juntas, and many other functions as well.  We conjecture that these functions are essentially the only functions efficiently isomorphism-testable.

To prove our main result, we also show that partial symmetry is efficiently testable. In turn, to prove
this result we had to revisit the junta testing problem.  We provide a new proof of correctness of the
nearly-optimal junta tester. Our new proof replaces the Fourier machinery of the original
proof with a purely combinatorial argument that exploits the connection between sets of variables
with low influence and intersecting families.

Another important ingredient in our proofs is a new notion of 
\emph{symmetric influence}.  We use this measure of influence to prove that partial symmetry is
efficiently testable and also to construct an efficient sample extractor for
partially symmetric functions.  We then combine the sample extractor with the testing-by-implicit-learning
approach to complete the proof that partially symmetric functions are efficiently isomorphism-testable.
\end{abstract}

\setcounter{page}{0}
\thispagestyle{empty}
\newpage

\section{Introduction}

Property testing considers the following general problem: given a property $\calP$, identify the minimum number of queries required to determine with high probability whether an input has the property $\calP$ or whether it is ÒfarÓ from $\calP$. This question was first formalized by Rubinfeld and Sudan~\cite{rubinfeld1996robust}.

\begin{definition}[\cite{rubinfeld1996robust}]
Let $\calP$ be a set of Boolean functions. An $\eps$-\emph{tester} for $\calP$ is a randomized algorithm which queries an unknown function $f: \{0,1\}^n\to \{0,1\}$ on a small number of inputs and
\vspace{-20pt}
 \begin{enumerate}[(i)]
   \setlength{\itemsep}{0pt}
   \item Accepts with probability at least $2/3$ when $f \in \calP$;
   \item Rejects with probability at least $2/3$ when $f$ is $\eps$-far from $\calP$,
 \end{enumerate}
 \vspace{-5pt}
where $f$ is $\eps$-\emph{far} from $\calP$ if $\dist(f,g) := | \{ x \in \{0,1\}^n \mid f(x) \neq g(x) \} | \geq \eps2^n$ holds for every $g \in \calP$.
\end{definition} 

Goldreich, Goldwasser, and Ron~\cite{goldreich1998property} extended the scope of this definition to graphs and other combinatorial objects. 
Since then, the field of property testing has been very
active.  For an overview of recent developments, we refer the reader to the 
surveys~\cite{ron2010algorithmic,rubinfeld2011survey} and the book~\cite{goldreich2010book}.

A notable achievement in the field of property testing is the complete characterization of graph properties that are testable with a constant number of queries~\cite{alon2009combinatorial}. An ambitious open problem is obtaining a similar characterization for properties of Boolean functions. Recently there has been a lot of progress on the restriction of this question to properties that are closed under linear or affine transformations~\cite{bhattacharyya2010unified,kaufman2008algebraic}. More generally, one might hope to settle this open problem for all properties of Boolean functions that are closed under relabeling of the input variables.

An important sub-problem of this open question is function isomorphism testing. Given a Boolean function $f$, the $f$-\emph{isomorphism testing} problem is to determine whether a function $g$ is isomorphic to $f$---that is, whether it is the same up to relabeling of the input variables---or \emph{far} from being so. A natural goal, and the focus of this paper, is to characterize the set of functions for which isomorphism testing can be done with a constant number of queries.

\paragraph{Previous work.}
The function isomorphism testing problem was first raised by Fischer et al.~\cite{fischer2004testing}. They observed that fully symmetric functions are trivially isomorphism testable with a constant number of queries. They also showed that every $k$-\emph{junta}, that is every function which depends on at most $k$ of the input variables, is isomorphism testable with $\poly(k)$ queries. This bound was recently improved by Chakraborty et al.~\cite{chakraborty2011nearly}, who showed that $O(k\log k)$ suffice. In particular, these results imply that juntas on a constant number of variables are isomorphism testable with a constant number of queries.

The first lower bound for isomorphism testing was also provided by Fischer et al.~\cite{fischer2004testing}. They showed that for small enough values of $k$, testing isomorphism to a $k$-linear function (i.e., a function that returns the parity of $k$ variables) requires $\Omega(\log k)$ queries.\footnote{More precisely, they showed that non-adaptive testers require $\tilde\Omega(\sqrt{k})$ queries. Here and in the rest of this section, tilde notation is used to hide logarithmic factors.} Following a series of recent works~\cite{goldreich2010testing,blais2011property,blais2011linear}, the exact query complexity for testing isomorphism to $k$-linear functions has been determined to be $\tilde\Theta(\min(k,n-k))$.

More general lower bounds for isomorphism testing were obtained by O'Donnell and the first author~\cite{blais2010lower}. In particular, they showed that testing isomorphism to \emph{any} $k$-junta that is \emph{far} from being a $(k-1)$-junta requires $\Omega(\log \log k)$ queries. This lower bound gives a large family of functions for which testing isomorphism requires a super-constant number of queries. Alon et al. have shown that in fact the query complexity of testing isomorphism is $\tilde\Theta(n)$ for almost every function~\cite{alon2011nearly} (see also~\cite{alon2010testing,chakraborty2011nearly}). 

\paragraph{Partially symmetric functions.}
As seen above, the only functions which we know are isomorphism testable with a constant number of queries are fully symmetric functions and juntas. Our motivation for the current work was to see if we can unify and generalize the results to encompass a larger class of functions. While symmetric functions and juntas may seem unrelated, there is in fact a strong connection. Symmetric functions, of course, are invariant under any relabeling of the input variables. Juntas satisfy a similar but slightly weaker invariance property. For every $k$-junta, there is a set of at least $n-k$ variables such that the function is invariant to any relabeling of these variables. Functions that satisfy this condition are called \emph{partially symmetric}.

\begin{definition}[Partially symmetric functions]
  For a subset $J \subseteq [n] := \{1,\ldots,n\}$, a function $f: \{0,1\}^n \to \{0,1\}$ is $J$-\emph{symmetric} if permuting the labels of the variables of $J$ does not change the function.
  Moreover, $f$ is called $t$-\emph{symmetric} if there exists $J \subseteq [n]$ of size at least $t$ such that $f$ is $J$-symmetric.
\end{definition}
Shannon first introduced partially symmetric functions as part of his investigation on the circuit 
complexity of Boolean functions~\cite{shannon1949synthesis}. He showed
that while most functions require an exponential number of gates to compute, every partially symmetric function can be implemented much more efficiently. Research on the role of partial symmetry in the 
complexity of implementing functions in circuits, binary decision diagrams, and 
other models has remained active ever since~\cite{das1971detecting,meinel1998book}. 
 Our results suggest that studying 
partially symmetric functions may also yield greater understanding of property testing
on Boolean functions.

\paragraph{Our results.}
The set of partially symmetric functions includes both juntas and symmetric functions, but the set also contains many other functions as well. A natural question is whether this entire class of functions is isomorphism testable with a constant number of queries. Our first main result gives an affirmative answer to this question.

\begin{theorem}\label{thm:psf-iso-test}
For every $(n-k)$-symmetric function $f:\{0,1\}^n\to\{0,1\}$ there exists an $\eps$-tester for $f$-isomorphism that performs $O(k \log k/\eps^2)$ queries.
\end{theorem}

A simple modification of an argument in Alon et al.~\cite{alon2011nearly} can be used to show that the bound in the above theorem is tight up to logarithmic factors. Indeed by this argument, testing isomorphism to almost every $(n-k)$-symmetric function requires $\Omega(k)$ queries.

We believe that the theorem might also be best possible in a different way. That is, we conjecture that the set of partially symmetric functions is essentially the set of functions for which testing isomorphism can be done with a constant number of queries. We discuss this conjecture with some supporting evidence in Section~\ref{sec:discussion}. 

The proof of our first main theorem follows the general outline of the proof that isomorphism testing to juntas can be done in a constant number of queries. The observation which allows us to make this connection is the fact that partially symmetric functions can be viewed as junta-like functions. More precisely, an $(n-k)$-symmetric function is a function that has $k$ special variables where for each assignment for these variables, the restricted function is fully symmetric on the remaining $n-k$ variables.

The proof for testing isomorphism of juntas has two main components. The first is an efficient junta testing algorithm. This enables us to reject functions that are far from being juntas. The second is a query efficient sampler of the ``core'' of the input function given that the function is close to a junta. The sampler can then be used in order to verify if the two juntas are indeed isomorphic. We generalize both of these components for partially symmetric functions.

\bigskip
Our second main result, and the first component of the isomorphism tester, is an efficient algorithm for testing partial symmetry.

\begin{theorem}\label{thm:psf-test}
The property of being $(n-k)$-symmetric for $k<n/10$ is testable with $O(\tfrac{k}{\eps} \log \tfrac{k}{\eps})$ queries.
\end{theorem}

The natural approach for proving this theorem is to try generalize the result on junta testing in~\cite{blais2009testing}. That result heavily relied on the notion of influence of variables. The \emph{influence} of a set $S$ of variables in a function $f$ is the probability that $f(x) \neq f(y)$ when $x$ is chosen uniformly at random and $y$ is obtained from $x$ by re-randomizing the values of $x_i$ for each $i \in S$. The notion of influence characterizes juntas: when $f$ is a $k$-junta, there is a set of size $n-k$ whose influence is $0$, whereas when $f$ is $\eps$-far from being a $k$-junta, every set of size $n-k$ has influence at least $\eps$.

We introduce a different notion of influence which we call \emph{symmetric influence}. The symmetric influence of a set $S$ of variables in $f$ is the probability that $f(x) \neq f(y)$ when $x$ is chosen uniformly at random and $y$ is obtained from $x$ by permuting the values of $\{x_i\}_{i\in S}$. This notion characterizes partially symmetric functions and satisfies several other useful properties. We provide the details in Section~\ref{sec:syminf}.

The proof of junta testing also relies on nice properties of the Fourier representation of the notion of influence. While symmetric influence has a clean Fourier representation, unfortunately it does not have the properties needed to carry over the proof in~\cite{blais2009testing} to the setting of partially symmetric functions. Instead, we must come up with a new proof technique.

Our proof of Theorem~\ref{thm:psf-test} uses a new connection to intersecting families. A family $\calF$ of subsets of $[n]$ is $t$-\emph{intersecting} if for every pair of sets $S,T \in \calF$, their intersection size is at least $|S \cap T| \geq t$.
This notion was introduced by Erd\H{o}s, Ko, and Rado and a sequence of works led to the complete characterization of the maximum size of $t$-intersecting families that contain sets of fixed size~\cite{erdos1961intersection,frankl1976erdos,wilson1984exact,ahlswede1997complete}. Dinur, Safra, and Friedgut recently extended those results to give bounds on the biased measure of intersecting families~\cite{dinur2005hardness, friedgut2008measure}.

Using results in intersecting families, we obtain a new and improved proof for the main lemma at the heart of the junta testing result~\cite{blais2009testing}. We describe the new proof and the connection to intersecting families in Section~\ref{sec:juntas}. Most importantly, the same technique can also be extended to complete the proof of Theorem~\ref{thm:psf-test}. We present this proof in Section~\ref{sec:psf}.

\bigskip
The second and final component of the isomorphism test for partially symmetric functions is an efficient way to sample
the core of such functions. An $(n-k)$-symmetric function $f$, which is symmetric over a set $J \subseteq [n]$ of size $|J| = n-k$,
has a concise representation as a function $f_{\core}:\{0,1\}^k \times \{0, 1,\ldots, n-k\} \to \{0,1\}$ which we call the \emph{core} of $f$.
The core is the restriction of $f$ to the variables in $\overline{J}$ (in the natural order),
with the additional Hamming weight of the variables in $J$.
To determine if two partially symmetric functions are isomorphic, it suffices to determine whether their
cores are isomorphic. We do so with the help of an efficient sample extractor.

\begin{definition}
A (1 query) $\delta$-\emph{sampler} for the $(n-k)$-symmetric function $f : \{0,1\}^n \to \{0,1\}$ 
is a randomized algorithm that queries $f$ on a single input and
returns a triplet $(x,w,z) \in \{0,1\}^k \times \{0,1,\ldots,n-k\} \times \{0,1\}$ where
\begin{itemize}
\setlength{\itemsep}{0pt}
\item The distribution of $(x,w)$ is $\delta$-close, in total variation distance, to $x$ being uniform over $\{0,1\}^k$
  and $w$ being binomial over $\{0,1,\ldots,n-k\}$ independently, and
\item $z = f_{\core}(x,w)$ with probability at least $1-\delta$.
\end{itemize}
\end{definition}

Our third main result is that for any $(n-k)$-symmetric function $f$, there is a query-efficient algorithm 
for constructing a $\delta$-sampler for $f$.

\begin{theorem}\label{thm:sampler}
Let $f : \{0,1\}^n \to \{0,1\}$ be $(n-k)$-symmetric with $k<n/10$. There is an algorithm that queries $f$
on $O(\tfrac{k}{\eta\delta} \log \tfrac {k}{\eta\delta})$ inputs and with probability at least $1- \eta$ outputs 
a $\delta$-sampler for $f$.
\end{theorem}

This theorem is a generalization of a recent result of Chakraborty et al.~\cite{chakraborty2011efficient}, who gave a similar construction for sampling the core of juntas. Their result has many applications related to testing by implicit learning~\cite{diakonikolas2007testing}. Our result may be of independent interest for similar such applications. We elaborate on this topic and present the proof of Theorem~\ref{thm:sampler} in Section~\ref{sec:psf-iso-test}.

\section{Intersecting families and testing juntas}
\label{sec:juntas}

We begin by revisiting the problem of junta testing. In this section, we give a new proof of the 
correctness of the $k$-junta tester first introduced in~\cite{blais2009testing}.  
At a high level, the junta tester is quite simple: it partitions the set
of indices into a large enough number of parts, then tries to identify all the 
parts that contain a relevant variable.
If at most $k$ such parts are found, the test accepts; otherwise it rejects.  The algorithm is described in
{\sc Junta-Test}.  (See~\cite{blais2009testing} for more details.)

\begin{algorithm}[tbh]
  \caption{\textsc{Junta-Test}$(f, k, \eps)$}
  \begin{algorithmic}[1]
    \STATE Create a random partition $\calI$ of the set $[n]$ into $r = \Theta(k^2)$ parts, and initialize $J = \emptyset$.
    \FOR{each $i = 1$ to $\Theta(k/\epsilon)$}
    \STATE Sample $x,y \in \{0,1\}^n$ uniformly at random.
    \IF{$f(x) \neq f(x_J y_{\overline{J}})$} \label{line:junta-test-neq}
    \STATE Use binary search to find a set $I \in \calI$ that contains a relevant variable.
    \STATE Set $J := J \cup I$.
    \STATE \textbf{if} $J$ is the union of $> k$ parts \textbf{then} reject.
    \ENDIF
    \ENDFOR
    \STATE Accept.
  \end{algorithmic}
\end{algorithm}

It is clear that the {\sc Junta-Test} always accepts $k$-juntas. The non-trivial part of the analysis 
involves showing that functions that are far from $k$-juntas are rejected by the tester with 
sufficiently high probability.  To do so, we must argue that the inequality in Step~\ref{line:junta-test-neq} is satisfied with
non-negligible probability whenever $f$ is far from $k$-juntas and $J$ is the union of at most $k$ parts.
This is accomplished by considering the {influence} of variables in a function. 

The \emph{influence} of the set $J \subseteq [n]$ of variables in the function $f : \{0,1\}^n \to \{0,1\}$ is
$$
\Inf_f(J) = \Pr_{x,y}[ f(x) \neq f(x_{\overline{J}} y_J) ]\ ,
$$
where $x_{\overline{J}} y_J$ is the vector $z \in \{0,1\}^n$ obtained by setting $z_i = y_i$ for every 
$i \in J$ and $z_i = x_i$ for every $i \in [n] \setminus J$.  By definition, the probability that the inequality in Step~\ref{line:junta-test-neq} is satisfied is exactly $\Inf_f(\overline{J})$.  To complete the analysis of correctness of the 
algorithm, we want to show that when $f$ is $\eps$-far from $k$-juntas with high probability over
the choice of the random partition $\calI$, if $J$ is the union of at most $k$ parts in $\calI$, then 
$\Inf(\overline{J}) \ge \frac{\eps}4$. We do so by exploiting only a couple basic facts about
the notion of influence.

\begin{lemma}[Fischer et al.~\cite{fischer2004testing}]
\label{lem:inf-fkrss}
For every $f : \{0,1\}^n \to \{0,1\}$ and every $J, K \subseteq [n]$,
$$
\Inf_f(J) \le \Inf_f(J \cup K) \le \Inf_f(J) + \Inf_f(K)\ .
$$
Furthermore, if $f$ is $\eps$-far from $k$-juntas and $|J| \le k$, then $\Inf_f(\overline{J}) \ge \eps.$
\end{lemma}


We also use the fact that the family of sets $J \subseteq [n]$ whose complements have small
influence form an intersecting family. For a fixed $t \ge 1$, a family $\calF$ of subsets of $[n]$ is called
$t$-\emph{intersecting} if any two sets $J$ and $K$ in $\calF$ have intersection size $|J \cap K| \ge t$. 
Much of the work in this area focused on bounding the size of $t$-intersecting families that 
contain only sets of a fixed size.  Dinur and Safra~\cite{dinur2005hardness} considered general
families and asked what the maximum $p$-\emph{biased measure} of such families can be.
For $0 < p < 1$, this measure is defined as $\mu_p(\calF) := \Pr_J[ J \in \calF ]$ where the probability
over $J$ is obtained by including each coordinate $i \in [n]$ in $J$ independently with probability $p$.
They showed that $2$-intersecting families have small $p$-biased measure~\cite{dinur2005hardness} and Friedgut showed how the same result also extends to $t$-intersecting families for $t > 2$~\cite{friedgut2008measure}.

\begin{theorem}[Dinur and Safra~\cite{dinur2005hardness}; Friedgut~\cite{friedgut2008measure}]\label{thm:DS}
Let $\calF$ be a $t$-intersecting family of subsets of $[n]$ for some $t \geq 1$. 
For any $p < \frac1{t+1}$, the $p$-biased measure of $\calF$ is bounded by
$
\mu_p(\calF) \le p^t.
$
\end{theorem}

We are now ready to complete the analysis of \textsc{Junta-Test}.
\begin{lemma}\label{lem:junta-main}
Let $f:\{0,1\}^n \to \{0,1\}$ be a function $\epsilon$-far from $k$-juntas and $\mathcal{I}$ be a random partition of $[n]$ into $r = c \cdot k^2$ parts, for some large enough constant $c$. Then with probability at least $5/6$, $\Inf_f(\overline{J}) \geq \epsilon / 4$ for any union $J$ of $k$ parts from $\calI$.
\end{lemma}

\begin{proof}
For $0 \leq t \leq \frac12$, let $\calF_t = \{J \subseteq [n] : \Inf_f(\overline{J}) < t\eps\}$ be the family of all 
sets whose complements have influence at most $t\eps$. For any two sets $J, K \in \calF_{1/2}$, the sub-additivity of influence implies that
$$
\Inf_f(\overline{J \cap K}) = \Inf_f(\overline{J} \cup \overline{K}) \le \Inf_f(\overline{J}) + \Inf_f(\overline{K}) < 2 \cdot \tfrac 12 \eps = \eps\ .
$$
But $f$ is $\eps$-far from $k$-juntas, so every set $S \subseteq [n]$ of size $|S| \le k$ satisfies 
$\Inf_f(\overline{S}) \ge \eps$.  Therefore, $|J \cap K| > k$ and, since this argument applies to
every pair of sets in the family, $\calF_{1/2}$ is a $(k+1)$-intersecting family.

Let us now consider two separate cases: when $\calF_{1/2}$ contains a set of size less than $2k$; and when it does not.  In the first case, let $J \in \calF_{1/2}$ be one of the sets of size $|J| < 2k$.  With high probability, the set $J$ is completely separated by the partition $\calI$. When this event occurs, then for every other set $K \in \calF_{1/2}$, $|J \cap K| \ge k+1$, which means that $K$ is not covered by any union of $k$ parts in $\calI$.  Therefore, with high probability $f$ is $\frac \eps 2$-far (and thus also $\frac \eps 4$-far) from $k$-part juntas with respect to $\calI$, as we wanted to show.

Consider now the case where $\calF_{1/2}$ contains only sets of size at least $2k$.  Then we claim that $\calF_{1/4}$ is a $2k$-intersecting family: otherwise, we could find sets $J, K \in \calF_{1/4}$ such that $|J \cap K| < 2k$ and $\Inf_f(\overline{J \cap K}) \le \Inf_f(\overline{J}) + \Inf_f(\overline{K}) < \frac \eps2$, contradicting our assumption.

Let $J \subseteq [n]$ be the union of $k$ parts in $\mathcal{I}$.  Since $\mathcal{I}$ is a random partition, $J$ is a random subset obtained by including each element of $[n]$ in $J$ independently with probability $p = \frac{k}{r} < \frac1{2k+1}$. By Theorem~\ref{thm:DS},
$$
\Pr_{\mathcal{I}}[ \Inf_f(\overline{J}) < \tfrac \eps 4] = \Pr[ J \in \calF_{1/4} ] = \mu_{k/r}(\calF_{1/4}) \le \left( k / r \right)^{2k}\ .
$$
Applying the union bound over the possible choices of $J$, we get that $f$ is $\frac \eps4$-close to a $k$-part junta with respect to $\calI$ with probability at most
\[
{r \choose k} \left( \frac k r \right)^{2k} \le \left( \frac{er}k \right)^{k} \left( \frac k r \right)^{2k} \le \left(\frac{ek}r\right)^{k} = O(k^{-k})\ . \qedhere
\]
\end{proof}

\section{Symmetric influence}
\label{sec:syminf}

The main focus of this paper is partially symmetric functions, that is, functions invariant under any reordering of the variables of some set $J \subseteq [n]$. Let $\calS_J$ denote the set of permutations of $[n]$ which only move elements from the set $J$. A function $f: \{0,1\}^n \to \{0,1\}$ is $J$-symmetric if $f(x) = f(\pi x)$ for every input $x$ and a permutation $\pi \in \calS_J$, where $\pi x$ is the vector whose $\pi(i)$-th coordinate is $x_i$.

For better analyzing partially symmetric functions, we introduce a new measure named \emph{symmetric influence}. The symmetric influence of a set measures how invariant the function is to reordering of the elements in that set.

\begin{definition}[Symmetric influence]
The \emph{symmetric influence} of a set $J \subseteq [n]$ of variables in a Boolean function $f : \{0,1\}^n \to \{0,1\}$ is defined as
$$
\SymInf_f(J) = \Pr_{x \in \{0,1\}^n, \pi \in \calS_J}[ f(x) \neq f( \pi x) ]\ .
$$
\end{definition}

It is not hard to see that in fact
a function $f$ is $t$-symmetric iff there exists a set $J$ of size $t$ such that $\SymInf_f(J) = 0$. A much stronger connection, however, exists between these properties as we will shortly describe.

Before showing some nice properties of symmetric influence, we mention that it also has a simple representation using Fourier coefficients of the function. Although we do not use the representation in this paper, we feel it might be of independent interest. See Appendix~\ref{appendix:syminf-via-fourier} for details.

\begin{lemma}\label{lem:distance-to-psf}
  Given a function $f: \{0,1\}^n \to \{0,1\}$ and
  a subset $J \subseteq [n]$,
  let $f_J$ be the $J$-symmetric function closest to $f$.
  Then, the symmetric influence of $J$ satisfies
  $$
  \dist(f, f_J) \leq \SymInf_f(J) \leq 2 \cdot \dist(f, f_J) \ .
  $$
\end{lemma}

\begin{proof}
For every weight $ 0 \leq w \leq n$ and $z \in \{0,1\}^{|\overline{J}|}$, define the layer $\layer{w}{J}{z} := \{x \in \{0,1\}^n \mid |x| = w \wedge x_{\overline{J}} = z\}$ to be the vectors of Hamming weight $w$ which identify with $z$ over the set $\overline{J}$
(where $|\layer{w}{J}{z}| = {|J| \choose w - |z|}$ if $|z| \leq w \leq |J| + |z|$ or 0 otherwise).
Let $p^w_z \in [0, \tfrac 12]$ be the fraction of the vectors in $\layer{w}{J}{z}$ one has to modify in order to make the restriction of $f$ over $\layer{w}{J}{z}$ constant.

With these notations, we can restate the definition of the symmetric influence of $J$ as follows.
\begin{eqnarray*}
\SymInf_f(J) 
&=&
  \sum_{z}\sum_w  \Pr_{x \in \{0,1\}^n} [x \in \layer{w}{J}{z}] \cdot \Pr_{x \in \{0,1\}^n, \pi \in \calS_J} [ f(x) \neq f(\pi x) \mid x \in \layer{w}{J}{z} ] \\
&=&
\frac{1}{2^n} \sum_z \sum_w |\layer{w}{J}{z}| \cdot  2p^w_z (1-p^w_z) \ .
\end{eqnarray*}
This holds as in each such layer, the probability that $x$ and $\pi x$ would result in two different outcomes is the probability that $x$ would be chosen out of the smaller part and $\pi x$ from the complement, or vise versa. 

The function $f_J$ can be obtained by modifying $f$ at $p^w_z$ fraction of the inputs in each layer $\layer{w}{J}{z}$, as each layer can be addressed separately and we want to modify as few inputs as possible. By this observation, we have the following equality.
$$
\dist(f, f_J) = \frac{1}{2^n} \sum_z \sum_w |\layer{w}{J}{z}| \cdot  p^w_z \ .
$$
But since $1 - p^w_z \in [\tfrac12,1]$, we have that $p^w_z \leq 2p^w_z(1-p^w_z) \leq 2p^w_z$ and therefore $\dist(f, f_J) \leq \SymInf_f(J) \leq 2 \cdot \dist(f, f_J)$ as required.
\end{proof}

\begin{corollary}\label{coro:syminf-psf}
  Let $f : \{0,1\}^n \to \{0,1\}$ be a function that is $\eps$-far from being $t$-symmetric.  Then for every set $J \subseteq [n]$ of size $|J| \geq t$, $\SymInf_f(J) \ge \eps$ holds.
\end{corollary}
\begin{proof}
  Fix $J \subseteq [n]$ of size $|J| \geq t$ and let $g$ be a $J$-symmetric function closest to $f$. Since $g$ is symmetric on any subset of $J$, it is in particular $t$-symmetric and therefore $\dist(f,g) \geq \eps$ as $f$ is $\eps$-far from being $t$-symmetric. Thus, by Lemma~\ref{lem:distance-to-psf}, $\SymInf_f(J) \geq \dist(f,g) \geq \eps$ holds.
\end{proof}

Corollary \ref{coro:syminf-psf} demonstrates the strong connection between symmetric influence and the distance from being partially symmetric, similar to the second part of Lemma~\ref{lem:inf-fkrss} for influence and juntas. The additional properties of influence used in Section \ref{sec:juntas} are monotonicity and sub-additivity (Lemma~\ref{lem:inf-fkrss}). The following lemmas show that the same
properties (approximately) hold for symmetric influence. See Appendices~\ref{appendix:syminf-monotonicity} and~\ref{appendix:syminf-subadditivity} for the proofs of both lemmas.

\begin{lemma}[Monotonicity]
\label{lem:syminf-monotonicity}
For any function $f : \{0,1\}^n \to \{0,1\}$ and any sets $J \subseteq K \subseteq [n]$,
$$
\SymInf_f(J) \le \SymInf_f(K)\ .
$$
\end{lemma}

\begin{lemma}[Weak sub-additivity]
\label{lem:syminf-weak-subadditivity}
There is a universal constant $c$ such that,
for any constant $0 < \gamma < 1$, a function $f : \{0,1\}^n \to \{0,1\}$, and sets $J, K \subseteq [n]$ of size at least $(1-\gamma) n$,
$$
\SymInf_f(J \cup K) \leq \SymInf_f(J) + \SymInf_f(K) + c \sqrt {\gamma}\ .
$$
\end{lemma}

\section{Testing partial symmetry}
\label{sec:psf} 

Let us now return to the problem of testing partial symmetry.  The goal of this section is to introduce
an efficient tester for this property by combining the ideas from Sections~\ref{sec:juntas} and~\ref{sec:syminf}.

We begin by introducing the testing algorithm \textsc{Partially-Symmetric-Test}. This algorithm is conceptually very
similar to the junta tester in Section~\ref{sec:juntas}.  Again, the main idea is to partition the 
variables into $O(k^2)$ parts and identify the parts that contain ``asymmetric'' variables. More
precisely, given a function $f : \{0,1\}^n \to \{0,1\}$, let us write $\core(f) \subseteq [n]$ to be
the maximum set $J$ of variables such that $f$ is $J$-symmetric.  We call the variables in
$\core(f)$ \emph{symmetric} and the variables in $[n] \setminus \core(f)$ are
called \emph{asymmetric}.  The function is $(n-k)$-symmetric iff it contains at most $k$ asymmetric 
variables. The algorithm exploits this characterization by trying to identify $k+1$ parts that contain
asymmetric variables.

\begin{algorithm}[tbh]
  \caption{\textsc{Partially-Symmetric-Test}$(f, k, \eps)$}
  \begin{algorithmic}[1]
    \STATE Create a random partition $\calI$ of $[n]$ into $r = \Theta(k^2/\epsilon^2)$ parts, and initialize $J:=\emptyset$.
    \STATE Pick a random workspace $W \in \calI$, and \textbf{if} $|W| < \frac{n}{2r}$ \textbf{then} fail.
    \FOR{each $i = 1$ to $\Theta(k/\epsilon)$}
    \STATE Let $I := $  \textsc{Find-Asymmetric-Set}$(f, \calI, J, W)$.
    \IF{$I \neq \emptyset$}
    \STATE Set $J := J \cup I$.
    \STATE \textbf{if} $J$ is the union of $> k$ parts \textbf{then} reject.
    \ENDIF
    \ENDFOR
    \STATE Accept.
  \end{algorithmic}
\end{algorithm}

There are two main differences in the analysis of \textsc{Partially-Symmetric-Test} and of
\textsc{Junta-Test} in Section~\ref{sec:juntas}.  The first is that we can no longer use a 
simple binary search algorithm to identify the parts that contain asymmetric variables, as we need to maintain the Hamming weight of our queries. To overcome
this challenge, we introduce the \textsc{Find-Asymmetric-Set} function, which satisfies the following
properties.

\begin{lemma}
\label{lem:properties-of-fas}
Let $f$ be a function, $\calI$ be a partition of $[n]$ into $r$ parts, $W\in \calI, |W| \geq \tfrac{n}{2r}$ be a workspace,
and $J$ be a union of parts from $\calI \setminus\{W\}$.
Then, there exists an algorithm \textsc{Find-Asymmetric-Set}$(f,\calI,J,W)$ which performs $O(\log r)$ queries such that
\begin{itemize}
\setlength{\itemsep}{0pt}
\item 
With probability $\SymInf_f(\overline{J})$, the algorithm returns a set $I \in \calI \setminus \{W\}$ disjoint to $J$; otherwise it returns $\emptyset$.
\item If $W$ has no asymmetric variable and $I \in \calI$ is returned, then $I$ has an asymmetric variable.
\end{itemize}
\end{lemma}
Due to space constraints, we provide a rough sketch of the algorithm and defer the details and analysis to Appendix~\ref{appendix:fas}. 
\textsc{Find-Asymmetric-Set} generates a random pair of $x \in \{0,1\}^n$ and $\pi \in \calS_{\overline{J}}$ and checks whether $f(x) \neq f(\pi x)$.
When this occurs, which happens with probability at least $\eps$ when $\SymInf_f(\overline{J}) \geq \epsilon$, we know there exists some asymmetric variable in $\overline{J}$. In order to identify a part $I \in \calI$, disjoint to $J$ and the workspace $W$, which contains an asymmetric variable we iteratively change $x$ to $\pi x$. 
In each step, we only \textit{permute} bits in one part $I \in \calI$ and the workspace $W$.
Since $f(x) \neq f(\pi x)$, we can find using binary search a set $I$, disjoint to $J$, such that permuting bits in $I \cup W$ changes the value of $f$.
By our assumption, $W$ has no asymmetric variables and therefore $I$ must contain such a variable.

The second and more important challenge in the analysis of \textsc{Partially-Symmetric-Test} is
the use of symmetric influence (rather than influence).
Similar to Lemma~\ref{lem:junta-main} for influence, we prove that if a function is far from being $(n-k)$-symmetric, then it is also far from being symmetric on any union of all but $k$ parts of a random partition (assuming it has enough parts).
The formal statement is given in Lemma~\ref{lem:psf-main}, where its proof follows a very similar technique to that of Lemma~\ref{lem:junta-main}.

\begin{lemma}\label{lem:psf-main}
  Let $f : \{0,1\}^n \to \{0,1\}$ be a function $\eps$-far from $(n-k)$-symmetric and $\calI$ be a random partition of $[n]$ into $r = c \cdot k^2/\eps^2$ parts, for some large enough constant $c$. Then with probability at least $8/9$, $\SymInf_f(\overline{J}) \ge \frac \eps{9}$ holds for any union $J$ of $k$ parts.
\end{lemma}

The main difference between this proof and the one of Lemma~\ref{lem:junta-main} arises from the weak sub-additivity of
symmetric influence. In light of this difference, our definition of families of sets whose complement has small symmetric influence\
includes only sets which are not too big. We use the observation that adding sets which contain elements of a family does not change
its existing intersection. In addition, due to the additive factor of the sub-additivity we prove a slightly weaker result
where the symmetric influence is at least $\eps/9$ and not $\eps/4$. The complete proof of Lemma~\ref{lem:psf-main} is deferred to Appendix~\ref{appendix:psf-main-proof}.

\bigskip
We can now complete the proof that partial symmetry is efficiently testable.

\begin{proof}[Proof of Theorem~\ref{thm:psf-test}]
  Note that $|W| \geq \frac{n}{2r}$ indeed holds with probability at least $8/9$ from Chernoff bound.
  By Lemma~\ref{lem:properties-of-fas}, \textsc{Find-Asymmetric-Set} performs $O(\log \frac{k}{\eps})$ queries
  according to our choice of $r$, and therefore the query complexity of \textsc{Partially-Symmetric-Test} is
  $O(\frac{k}{\epsilon} \log \frac{k}{\epsilon})$. 

  Suppose $f$ is an $(n-k)$-symmetric function.
  The probability that $W$ contains an asymmetric variable is at most $k/r \leq 2/9$.
  Conditioned this did not occur, every set returned by \textsc{Find-Asymmetric-Set} contains an asymmetric variable.
  Since there are at most $k$ such variables, $J$ would be the union of at most $k$ sets and we would accept.

  Suppose $f$ is a function $\epsilon$-far from being $(n-k)$-symmetric.
  From Lemma~\ref{lem:psf-main},
  with probability at least $8/9$,
  $\SymInf_f(\overline{J}) \geq \epsilon/9$ holds while $J$ consists of at most $k$ parts.
  Conditioned on that, by executing \textsc{Find-Asymmetric-Set} $O(k/\epsilon)$ times 
  we obtain more than $k$ parts with probability at least $8/9$, according to Lemma~\ref{lem:properties-of-fas}.
  Thus, we reject with probability at least $2/3$.
\end{proof}

\section{Isomorphism testing of partially symmetric functions}
\label{sec:psf-iso-test}

In this section we prove that isomorphism testing of partially symmetric functions can be done with a constant number of queries. The algorithm we describe consists of two main components, and follow a similar approach to the one used in~\cite{chakraborty2011nearly} when they showed juntas are isomorphism testable. The first, which we already described in Section~\ref{sec:psf}, is an efficient tester for the property of being partially symmetric. Once we know the input function is indeed close to being partially symmetric, we can verify it is isomorphic (or at least very close) to the correct one. 
The second component of the algorithm is therefore an efficient sampler from the \emph{core} of a function which is (close to) partially symmetric. Comparing the cores of two partially symmetric functions suffices to identify if two such functions are isomorphic or far from it.

Ideally, when sampling the core of a partially symmetric function $f$, we would like to sample it according to the marginal distribution of sampling $f$ at a uniform input $x \in \{0,1\}^n$. We denote this marginal distribution over $\{0,1\}^k \times \{0,1, \ldots, n-k\} $ by $\calD^*_{k,n}$, which is in fact uniform over $\{0,1\}^k$ and binomial over $\{0,1, \ldots, n-k\}$, independently.

In our scenario, sampling the core of a function according to this distribution is not possible since we do not know the exact location of all the $k$ asymmetric variables. Instead, we use the knowledge discovered by the partial symmetry tester, i.e., sets with asymmetric variables. Given these sets, we are able to define a sampling distribution over $\{0,1\}^n$ such that we know the input of the core for each query, and whose marginal distribution over the core is close enough to $\calD^*_{k,n}$.

\begin{definition}
Let $\calI$ be some partition of $[n]$ into an odd number of parts and let $W \in \calI$ be the workspace.
Define the distribution $\calDIW$ over $\{0,1\}^n$ to be as follows. 
Pick a random Hamming weight $w$ according to the binomial distribution over $\{0, \ldots, n\}$ and output, if exists, a random $x \in \{0,1\}^n$ of Hamming weight $|x| = w$ such that for every $I \in \calI \setminus \{W\}$, either $x_I \equiv 0$ or $x_I \equiv 1$. 
When no such $x$ exists, return the all zeros vector.
\end{definition}

The sampling distribution which we just defined, together with the random choice of the partition and workspace, satisfies the following two important properties. The first, being close to uniform over the inputs of the function.
The second, having a marginal distribution over the core of a partially symmetric function close to $\calD^*_{k,n}$.
These properties are formally written here as Proposition~\ref{prop:samples}, whose proof
is rather technical and appears in Appendix~\ref{appendix:dist-diw}.

\begin{proposition}\label{prop:samples}
Let  $J=\{j_1, \ldots, j_k\} \subseteq [n]$ be a set of size $k$, and $r = \Omega(k^2)$ be odd.
If $x \sim \calDIW$ for a random partition $\calI$ of $[n]$ into $r$ parts and a random workspace $W \in \calI$, then
\vspace{-5pt}
\begin{itemize}
\setlength{\itemsep}{-2pt}
\item $x$ is $o(1/n)$-close to being uniform over $\{0, 1\}^n$, and
\item $(x_J, |x_{\overline{J}}|)$ is $c/k$-close to being distributed according to $\calD^*_{k,n}$, for our choice of $0 < c < 1$.
\end{itemize}
\end{proposition}

We are now ready to describe the algorithm for isomorphism testing of $(n-k)$-symmetric functions.
Given an $(n-k)$-symmetric function $f$, the algorithm tests
whether the input function $g$ is isomorphic to $f$ or $\eps$-far from being so.

\begin{algorithm}
  \caption{\textsc{Partially-Symmetric-Isomorphism-Test}$(f, k, g, \eps)$}
  \begin{algorithmic}[1]
    \STATE Perform \textsc{Partially-Symmetric-Test}$(g,k, \eps/1000)$ and reject if failed.
    \STATE Let $\calI$ and $W \in \calI$ be the partition and workspace used by the algorithm.
    \STATE Let $J$ be the union of the $k$ parts identified by the algorithm.
    \FOR{each $i = 1$ to $\Theta(k \log k/\epsilon^2)$}
      \STATE Query $g(x)$ at a random $x \sim \calDIW$
    \ENDFOR
    \STATE Accept iff $(1-\eps/2)$-fraction of the queries are consistent with some isomorphism $f_{\pi}$ of $f$, 
        which maps the asymmetric variables of $f$ into all $k$ parts of $J$.
  \end{algorithmic}
\end{algorithm}

We provide here a sketch of the analysis of the algorithm. See Appendix~\ref{appendix:iso-psf-proof} for the formal analysis and complete proof of Theorem~\ref{thm:psf-iso-test}.
The first case to analyze is when $g$ is rejected by \textsc{Partially-Symmetric-Test}, which implies that with
good probability it is not $(n-k)$-symmetric and in particular not isomorphic to $f$. 
Assume now that \textsc{Partially-Symmetric-Test} did not reject and therefore $g$ is likely to be $\eps/1000$-close to being
$(n-k)$-partially symmetric.
Let $\calI, W$ and $J$ be the partition, workspace and union of $k$ parts identified by the algorithm.
The main idea of the proof is showing that with good probability, there exists a function $h$ that (a)
is $\eps/250$-close to $g$, and (b)
is $(n-k)$-symmetric with asymmetric variables contained in $J$ and separated by $\calI$.
We prove the existence of this function $h$ using the properties of symmetric influence presented in Section~\ref{sec:psf}.
Assuming such $h$ exists, we use Proposition~\ref{prop:samples} in order to
show that our queries to $g$, according to the sampling distribution, are in fact $\eps/10$-close to querying $h$'s core.

We now consider the following two cases. If $g$ is isomorphic to $f$, then for some isomorphism $f_{\pi}$ of $f$, which maps the
asymmetric variables of $f$ into the parts of $J$, it holds that
$\dist(f_{\pi}, h) \leq \dist(f_{\pi}, g) + \dist(g,h) \leq \eps/500 + \eps/250$. Notice that we cannot assume that $g = f_{\pi}$ as it
is possible that one of the asymmetric variables of $g$ are not in $J$ (but the distance must be small).
If $g$ was $\eps$-far from being isomorphic to $f$, then for every isomorphism $f_\pi$ of $f$, 
$$
\dist(f_{\pi}, h) \geq \dist(f_{\pi}, g) - \dist(g,h) \geq \eps - \eps/250 \ .
$$
Given that there are only $k!$ isomorphisms of $f$ we need to consider, performing $\Theta(k \log k/\epsilon^2)$
queries suffices for returning the correct answer in both cases, with good probability.

\bigskip
As we outlined above, we in fact build an efficient sampler for the core of $(n-k)$-symmetric functions (or functions close to being so). Given the parts identified by \textsc{Partially-Symmetric-Test}, assuming it did not reject, we can sample the function's core by querying it at a single location, where the distribution over the core's inputs is close to $\calD^*_{k,n}$. The algorithm and proof of Theorem~\ref{thm:sampler} are deferred to Appendix~\ref{appendix:sampler}.

\section{Discussion}
\label{sec:discussion}

We showed that every partially symmetric function is isomorphism testable with a constant number of queries.  It's easy to see that functions that are ``close'' to partially symmetric can also be isomorphism-tested with a constant number of queries.  We believe that our result not only unifies the previous classes of functions efficiently isomorphism-testable, but that it includes essentially \emph{all} of these functions.

\begin{conjecture}
\label{conj:lb}
Let $f : \{0,1\}^n \to \{0,1\}$ be $\eps$-far from $(n-k)$-symmetric.  Then testing $f$-isomorphism requires at least $\Omega(\log \log k)$ queries.
\end{conjecture}

In fact, we believe that more is true---perhaps even $\Omega(k)$ queries are required.  But the weaker bound (or, indeed, any function that grows with $k$) is sufficient to complete the qualitative characterization of functions that are isomorphism-testable with a constant number of queries.

The known hardness results on isomorphism testing are all consistent with Conjecture~\ref{conj:lb}.
In particular, by the result in~\cite{alon2011nearly}, we know
that testing $f$-isomorphism requires at least $\Omega(k)$ queries for \emph{almost all} functions
$f$ that are $\eps$-far from $(n-k)$-symmetric.  A simple extension of the proof in~\cite{blais2010lower} shows that for every $(n-k)$-symmetric function $f$ that is $\eps$-far from $(n-k+1)$-symmetric, testing $f$-isomorphism requires
$\Omega(\log \log k)$ queries (assuming $k/n$ is bounded away from 1).

Lastly, let us consider another natural definition of
partial symmetry that encompasses both symmetric functions and juntas. The function
 $f : \{0,1\}^n \to \{0,1\}$ is $k$-\emph{part symmetric} if there is a partition 
 $\calI = \{I_1,\ldots,I_k\}$ of $[n]$ such that
$f$ is invariant under any permutation $\pi$ of $[n]$ where $\pi(I_i) = I_i$ for every $i = 1,\ldots,k$. One may be tempted to guess that $k$-part symmetric functions are efficiently isomorphism-testable. That is not the case, even when $k = 2$.  To see this, consider
the function $f(x) = x_1 \oplus x_2 \oplus \cdots \oplus x_{n/2}$. This function is $2$-part symmetric,
but testing isomorphism to $f$ requires $\Omega(n)$ queries~\cite{blais2011property}.

\section*{Acknowledgments}

We thank Noga Alon, Per Austrin, Irit Dinur, Ehud Friedgut, and Ryan O'Donnell for useful discussions and valuable feedback.

\bibliographystyle{plain}
\bibliography{TestingPSFs}

\appendix

\section{Properties of symmetric influence}

\subsection{Fourier representation of symmetric influence}
\label{appendix:syminf-via-fourier}

For convenience, we consider functions whose ranges are $\{-1,1\}$ instead of $\{0,1\}$.
Then, the symmetric influence of a function can be expressed as follows.
\begin{proposition}\label{prop:syminf-fourier}
Given a Boolean function $f: \{0,1\}^n \to \{-1,1\}$ and a set $J \subseteq [n]$,
the symmetric influence of $J$ with respect to $f$ can also be computed as
$$
\SymInf_f(J) = \tfrac{1}{2} \sum_{S \subseteq [n]} \Var_{\pi \in \calS_J} [ {\widehat{f}(\pi S)} ]
$$
where $\widehat{f}(S)$ is the Fourier coefficient of $f$ for the set $S \subseteq [n]$, and $\pi S = \{ \pi(i) \mid i \in S \}$.
\end{proposition}

The proposition indicates that the symmetric influence of any set $J$ can be computed as a function of the variance of the
Fourier coefficients of the function in the different layers. Each layer here refer to all the Fourier coefficients
of sets which share the intersection with $[n] \setminus J$ and the intersection size with $J$,
resulting in $(|J|+1)2^{n-|J|}$ different layers.

\bigskip
The key to proving this proposition is the following basic result on linear functions.
Recall that for a set $S \subseteq [n]$, the function $\chi_S : \{0,1\}^n \to \{-1,1\}$ is
defined by $\chi_S(x) = (-1)^{\sum_{i \in S} x_i}$.

\begin{lemma}\label{lem:syminf-fourier-mult}
Fix $J, S, T \subseteq [n]$. Then
$$
\E_{x \in \{0,1\}^n, \pi \in \calS_J}[ \chi_S(x) \cdot \chi_T(\pi x) ] =
\begin{cases}
{|J| \choose |S \cap J|}^{-1} & \mbox{if } \exists \pi \in \calS_J, \ \pi S = T \\
\ \ \ 0 & \mbox{otherwise.}
\end{cases}
$$
\end{lemma}

\begin{proof}
For any vector $x \in \{0,1\}^n$, any set $S \subseteq [n]$, and any permutation $\pi \in \calS_n$, we have the identity
$\chi_S(\pi x) = \chi_{\pi^{-1} S}(x)$.
So
$$
\E_{x \in \{0,1\}^n,\pi \in \calS_J}[ \chi_S(x) \cdot \chi_T(\pi x) ] =
  \E_{x,\pi}[ \chi_S(x) \chi_{\pi^{-1} T}(x) ] = \E_\pi\left[ \E_x[ \chi_S(x) \chi_{\pi^{-1}T}(x)]\right]\ .
$$
But $\mathbf{E}_x[ \chi_S(x) \chi_{\pi^{-1}T}(x)] = \mathbf{1}[ S = \pi^{-1} T]$, so we also have
$$
\E_{x \in \{0,1\}^n,\pi \in \calS_J}[ \chi_S(x) \cdot \chi_T(\pi x) ] = \Pr_{\pi \in \calS_J}[ S = \pi^{-1}T] = \Pr_{\pi \in \calS_J}[\pi S = T]\ .
$$
The identity $\pi S= T$ holds iff the permutation $\pi$ satisfies $\pi(i) \in T$ for every $i \in S$.
Since we only permute elements from $J$, the sets $S$ and $T$ must agree on the elements of $[n] \setminus J$.
If this is not the case, or if the intersection of the sets with $J$ is not of the same size, no such permutation exists.
Otherwise, this event occurs if the elements of $S \cap J$ are mapped to the exact locations of $T \cap J$.
This holds for one out of the ${|J| \choose |S \cap J|}$ possible sets of locations, each with equal probability.
\end{proof}

\begin{proof}[Proof of Proposition~\ref{prop:syminf-fourier}]
By appealing to the fact that $f$ is $\{-1,1\}$-valued, we have that
$$
\Pr_{x,\pi}[ f(x) \neq f(\pi x)] = \frac{1}{4}\E_{x, \pi}[ f(x)^2 + f(\pi x)^2 - 2 f(x) f(\pi x)].
$$
Applying linearity of expectation and Parseval's identity, we obtain
$$
\E_{x, \pi}[ f(x)^2 + f(\pi x)^2 - 2 f(x) f(\pi x)] =
  2 \sum_{S \subseteq [n]} \hat{f}(S)^2 - 2 \sum_{S,T \subseteq [n]} \hat{f}(S) \hat{f}(T) \E_{x,\pi}[ \chi_S(x) \chi_T(\pi x)]\ .
$$
Fix any $S \subseteq [n]$.  By Lemma~\ref{lem:syminf-fourier-mult},
$$
\sum_{T \subseteq [n]} \hat{f}(T) \E_{x,\pi}[ \chi_S(x) \chi_T(\pi x)] =
  \sum_{\pi \in \calS_J } \frac{\hat{f}(\pi S)}{{|J| \choose |S \cap J|}} = \E_{\pi \in \calS_J }[ \hat{f}(\pi S) ]\ .
$$
Given this equality,
$$
\sum_{S,T \subseteq [n]} \hat{f}(S) \hat{f}(T) \E_{x,\pi}[ \chi_S(x) \chi_T(\pi x)] =
  \sum_S \hat{f}(S) \E_{\pi \in \calS_J}[ \hat{f}(\pi S)]\ .
$$
By applying some elementary manipulation, we now get
\begin{eqnarray*}
\Pr_{x,\pi}[ f(x) \neq f(\pi x)]
 & = & \frac{1}{2} \sum_{S} \hat{f}(S) \left(\hat{f}(S) -  \E_{\pi}[ \hat{f}(\pi S)]\right) = \\
 & = & \frac{1}{2} \sum_{S} (\E_{\pi}[ \hat{f}(\pi S)^2] - \E_{\pi}[ \hat{f}(\pi S)]^2) = \frac{1}{2}\sum_S \Var_\pi[ \hat{f}(\pi S)]\ . 
\end{eqnarray*}
\end{proof}

\subsection{Monotonicity of symmetric influence}
\label{appendix:syminf-monotonicity}
\newtheorem*{lemmonotone}{Lemma~\ref{lem:syminf-monotonicity}}
\begin{lemmonotone}[Restated]
For any function $f : \{0,1\}^n \to \{0,1\}$ and any sets $J \subseteq K \subseteq [n]$,
$$
\SymInf_f(J) \le \SymInf_f(K)\ .
$$
\end{lemmonotone}

\begin{proof}
Fix a function $f$ and two sets $J,K \subseteq [n]$ so that $J \subseteq K$. We have seen before that the
symmetric influence can be computed in layers, where each layer
is determined by the Hamming weight and the elements outside the set we are considering. 
Using the fact that $\Var(X) = \Pr[X=0]\cdot \Pr[X=1]$,
the symmetric influence is twice the expected variance over all the
layers (considering also the size of the layers). Using the same notation as before,
\begin{eqnarray*}
  \SymInf_f(J) & = & \frac{1}{2^n} \sum_{z} \sum_{w} |\layer{w}{J}{z}|\cdot 2 \Var_{x}
  [f(x) \mid x \in \layer{w}{J}{z}] \\
  & = & 2\cdot \E_{y} \left[
    \Var_{x} [f(x) \mid x \in \layer{|y|}{J}{y_{\overline{J}}}] \right] \ .
\end{eqnarray*}

A key observation is that since $\overline{K} \subseteq \overline{J}$,
the layers determined when considering $J$ are a refinement of the
layers determined when considering $K$. Together with the fact that
$\Var(X) = \Pr[X=0]\cdot \Pr[X=1]$ is a concave function in the range
$[0,1]$, we can apply Jensen's inequality on each layer before and
after the refinement to get the desired inequality. More precisely,
for every $z\in\{0,1\}^{|\overline{K}|}$ and $0 \leq w \leq n$,
$$
\Var_{x} [f(x) \mid x \in \layer{w}{K}{z} ] \geq \E_{y}
\left[ \Var_{x} [f(x) \mid x \in \layer{w}{J}{y_{\overline{J}}} ] \mid
y \in \layer{w}{K}{z} \right] \ .
$$
Averaging this over all layers, we get the desired result.
\end{proof}

\subsection{Weak sub-additivity of symmetric influence}
\label{appendix:syminf-subadditivity}

In this section we prove that symmetric influence satisfies weak sub-additivity. It might be tempting to think that
strong sub-additivity holds, as in the standard notion of influence, however this is not the case.
For example, consider the function $f(x) = f_1(x_J) \oplus f_2(x_K)$ for some partition $[n] = J \cup K$ and two randomly
chosen symmetric functions $f_1, f_2$. Since $f$ is far from symmetric, $\SymInf_f([n]) = \SymInf_f(J \cup K) > 0$
while $\SymInf_f(J) = \SymInf_f(K) = 0$.

The additive factor of $c\sqrt{\gamma}$ in Lemma~\ref{lem:syminf-weak-subadditivity} is derived from the
distance between the two distributions $\pi_{J \cup K} x$ and $\pi_J \pi_K x$, for a random $x \in \{0,1\}^n$
and random permutations from $\calS_{J\cup K}, \calS_J, \calS_K$. When the sets $J$ and $K$ are large,
the distance between these distributions is relatively small which therefore result in this weak sub-additivity property.

The analysis of the lemma is done using hypergeometric distributions, and the distance between them.
Let $\calH_{n,m,k}$ be the hypergeometric distribution obtained when we pick $k$ balls out of $n$,
$m$ of which are red, and count the number of red balls we obtained.
Let $\dtv(\cdot,\cdot)$ denote the statistical distance between two distributions.
The following two lemmas would be useful for our proof.

\begin{lemma}\label{lem:hypergeometric}
  Let $J,K \subseteq [n] $ be two sets and $\pi, \pi_J, \pi_K$ be permutations chosen uniformly at random from $\calS_{J \cup K}, \calS_J,\calS_K$, respectively.
  For a fixed $x \in \{0,1\}^n$, 
  we define $\calD_{\pi x}$ and $D_{\pi_J \pi_K x}$ as the distribution of $\pi x$ and $\pi_J \pi_K x$, respectively.
  Then, 
  \begin{eqnarray*}
    \dtv(D_{\pi x}, D_{\pi_J \pi_K x}) = \dtv(\calH_{|J \cup K|,|x_{J \cup K}|,|K \setminus J|}, \calH_{|K|,|x_K|, |K \setminus J|})
  \end{eqnarray*}
  holds.
\end{lemma}

\begin{lemma}\label{lem:dtv}
  Let $n,m,n',m',k$ be non-negative integers with $k,n' \leq \gamma n$ for some $\gamma \leq \frac{1}{2}$.
  Suppose that $| m -  \frac{n}{2} | \leq t \sqrt{n}$ and $|m' - \frac{n'}{2}| \leq t \sqrt{n'}$ hold
  for some $t \leq \frac{1}{100 \sqrt{\gamma} }$.
  Then, 
  \begin{eqnarray*}
    \dtv(\calH_{n,m,k},\calH_{n-n',m-m',k}) \leq c_{\ref{lem:dtv}} (1 + t) \gamma\ .
  \end{eqnarray*}
  holds for some universal constant $c_{\ref{lem:dtv}}$.
\end{lemma}

We first show how these lemmas imply the proof of Lemma~\ref{lem:syminf-weak-subadditivity},
and will afterwards prove them.

\newtheorem*{lemweak}{Lemma~\ref{lem:syminf-weak-subadditivity}}
\begin{lemweak}[Restated]
There is a universal constant $c$ such that,
for any constant $0 < \gamma < 1$, a function $f : \{0,1\}^n \to \{0,1\}$ and sets $J, K \subseteq [n]$ of size at least $(1-\gamma) n$,
$$
\SymInf_f(J \cup K) \leq \SymInf_f(J) + \SymInf_f(K) + c \sqrt {\gamma}\ .
$$
\end{lemweak}

\begin{proof}
  Let $\pi, \pi_J$ and $\pi_K$ be as in Lemma~\ref{lem:hypergeometric}
  and fix $x \in \{0,1\}^n$ to be some input.
  \begin{eqnarray*}
    \Pr_{\pi }[f(x) \neq f(\pi x)]
    &\leq& 
    \Pr_{\pi_J,\pi_K}[f(x) \neq f(\pi_J \pi_K x)] + \dtv(\calD_{\pi x},\calD_{\pi_J \pi_K x})  \\
    &\leq&
    \Pr_{\pi_K}[f(x) \neq f(\pi_K x)] + \Pr_{\pi_J,\pi_K}[f(\pi_K x) \neq f(\pi_J\pi_K x)] + \dtv(\calD_{\pi x},\calD_{\pi_J \pi_K x})
  \end{eqnarray*}
  By summing over all possible inputs $x$ we have
  \begin{eqnarray*}
    \SymInf_f(J \cup K)
    &=& 
    \Pr_{x,\pi}[f(x) \neq f(\pi x)]  = \frac{1}{2^n}\sum_{x}\Pr_{\pi }[f(x) \neq f(\pi x)]  \\
    &\leq&
    \SymInf_f(J) + \SymInf_f(K) + \frac{1}{2^n}\sum_{x} \dtv(\calD_{\pi x},\calD_{\pi_J \pi_K x}) \ .
  \end{eqnarray*}
  By applying Lemma~\ref{lem:hypergeometric} over each input $x$, it suffices to show that
  \begin{eqnarray}
  \label{eq:distance}
  \frac{1}{2^n}\sum_{x} \dtv(\calD_{\pi x},\calD_{\pi_J \pi_K x}) = 
  \frac{1}{2^n}\sum_{x} \dtv(\calH_{|J \cup K| ,|x_{J \cup K}|,|K \setminus J |},\calH_{|K|, |x_{K}| , |K \setminus J|}) \leq
  c \sqrt{\gamma} \ .
  \end{eqnarray}

  Ideally, we would like to apply Lemma~\ref{lem:dtv} on every input $x$ and get the desired result, however this
  is not possible as some inputs does not satisfy the requirements of the lemma.
  Therefore, we perform a slightly more careful analysis.
  Let us choose $c \geq 2$ and assume $\gamma \leq \tfrac14$ (as otherwise the claim trivially holds).
  Fix $\gamma' = \gamma/(1-\gamma) \leq \tfrac12$ and $t = \tfrac{1}{100\sqrt{\gamma'}}$.
  We first note that regardless of $x$, the required conditions on the size of the sets hold.
  To be exact, $|J \setminus K| \leq \gamma' |J \cup K|$ and $|K \setminus J| \leq \gamma' |J \cup K|$
  since $|J \cup K| \geq (1-\gamma) n$ and $|J \setminus K| \leq |\overline{K}| \leq \gamma n$
  (and similarly $|K \setminus J| \leq \gamma n$).
  
  We say an input $x$ is \emph{good} if it satisfies the other conditions of Lemma~\ref{lem:dtv}.
  That is, both $\left| |x_{J \cup K}| - \frac{|J \cup K|}{2} \right| \leq t \sqrt{|J \cup K|}$
  and $\left| |x_{J \setminus K}| - \frac{|J \setminus K|}{2} \right| \leq t \sqrt{|J \setminus K|}$ hold.
  Otherwise we call such $x$ \emph{bad}.
  From the Chernoff bound and the union bound, the probability that $x$ is bad is at most
  $4\exp(-2t^2) \leq 4 \exp\left(-\frac{1}{5000 \gamma'}\right) \leq c'\gamma$ for some constant $c'$
  (notice that $\gamma' \leq 2\gamma$).
  
  By applying Lemma~\ref{lem:dtv} over the good inputs we get
  $$
  \eqref{eq:distance} \leq \frac{1}{2^n}\sum_{x:bad} 1 + \frac{1}{2^n}\sum_{x:good}c_{\ref{lem:dtv}}(1+t)\gamma
  \leq c'\gamma + c_{\ref{lem:dtv}}(1+t)\gamma \leq c\sqrt{\gamma}
  $$
  for some constant $c$, as required.
\end{proof}

\begin{proof}[Proof of Lemma~\ref{lem:hypergeometric}]
Since both distributions $D_{\pi x}$ and $D_{\pi_J \pi_K x}$ only modify coordinates in $J \cup K$, we can ignore all other coordinates.
Moreover, it is in fact suffices to look only at the number of ones in the coordinates of $K \setminus J$ and $J \cup K$,
which completely determines the distributions.
Let $D_z$ denote the uniform distribution over all elements $y\in \{0,1\}^n$ such that
$|y| = |x|$, $y_{\overline{J \cup K}} = x_{\overline{J \cup K}}$ and $|y_{K \setminus J}| = z$ 
(which also fixes the number of ones in $y_J$).
Notice that this is well defined only for values of $z$ such that
$\max \{ 0, |x_{J \cup K}| - |J| \} \leq z \leq \min \{  |x_{J \cup K}| , |K \setminus J| \} $.

Given this notation, $D_{\pi x}$ can be looked at as 
choosing $z \sim \calH_{|J \cup K|,|x_{J \cup K}|,|K \setminus J|}$ and returning $y \sim D_z$.
This is because we apply a random permutation over all elements of $J \cup K$,
and therefore the number of ones inside $K \setminus J$ is indeed distributed like $z$.
Moreover, the order inside both sets $K \setminus J$ and $J$ is uniform.

The distribution $D_{\pi_J \pi_K x}$ can be looked at as choosing
$z \sim \calH_{|K|,|x_K|, |K \setminus J|}$ and returning $y \sim D_z$.
The number of ones in $K \setminus J$ is determined already after applying $\pi_K$.
It is distributed like $z$ as we care about the choice of 
$|K \setminus J|$ out of the $|K|$ elements, and $|x_K|$ of them are ones (and their order is uniform).
Later, we apply a random permutation $\pi_J$ over all other relevant coordinates,
so the order of elements in $J$ is also uniform.

Since the distributions $D_z$ are disjoint for different values of $z$, this implies that the distance between the two
distributions $D_{\pi x}$ and $D_{\pi_J \pi_K x}$ depends only on the number of ones chosen to be inside $K \setminus J$.
Therefore we have
$$
 \dtv(D_{\pi x}, D_{\pi_J \pi_K x}) = \dtv(\calH_{|J \cup K|,|x_{J \cup K}|,|K \setminus J|}, \calH_{|K|,|x_K|, |K \setminus J|})
$$
as required.
\end{proof}

\begin{proof}[Proof of Lemma~\ref{lem:dtv}]
Our proof uses the connection between hypergeometric distribution and the binomial distribution, which we denote by $\calB_{n,p}$ (for $n$ experiments, each with success probability $p$). By the triangle inequality we know that
  \begin{eqnarray}
    \dtv(\calH_{n,m,k},\calH_{n-n',m-m',k})
    \leq
    \dtv(\calH_{n,m,k},\calB_{k,p}) + \dtv(\calB_{k,p},\calB_{k,p'}) + \dtv(\calB_{k,p'}, \calH_{n-n',m-m',k})  \label{eq:first}
  \end{eqnarray}
  where $p = \frac{m}{n}$ and $p' = \frac{m-m'}{n-n'}$.
  In order to bound the distances we just introduced, we use  the following two lemmas.
\begin{lemma}[Example~1 in~\cite{soon1996binomial}]
\label{lem:h-and-b}
  $\dtv(\calH_{n,m,k},\calB_{k,p}) \leq \frac{k}{n}$  holds for $p = \frac{m}{n}$.
\end{lemma}

\begin{lemma}[\cite{adell2006exact}\label{lem:b-and-b}]
  Let $0 < p < 1$ and $0 < \delta < 1-p$.
  Then,
  \begin{eqnarray*}
    \dtv(\calB_{n,p},\calB_{n,p+\delta}) \leq \frac{\sqrt{e}}{2}\frac{\tau_{n,p}(\delta)}{(1-\tau_{n,p}(\delta))^2}
  \end{eqnarray*}
  provided $\tau_{n,p}(\delta) = \delta\sqrt{\frac{n+2}{2p(1-p)}} < 1$.
\end{lemma}

  Before using the above lemmas, we analyze some of the parameters.
  First, when $k=0$ the lemma trivially holds and we therefore assume $k\geq 1$.
  Notice that this implies that $n\gamma \geq k \geq 1$.
  The probability $p$ is known to be relatively close to half. To be exact,
  $|p - \tfrac12| \leq t\sqrt{n}/n \leq \tfrac{1}{100\sqrt{n\gamma}} \leq \tfrac{1}{100}$
  and therefore $\tfrac{1}{p(1-p)} < 6$.
  Assume $p \leq p'$ and let $\delta = p' - p$ (the other case can be treated in the same manner).
  We first bound $\delta$ as follows.
  \begin{eqnarray*}
    \delta &=&
    \frac{mn' - nm'}{n(n-n')}
    \leq
    \frac{1}{n(n-n')} \left( \left(\frac{n}{2}+t \sqrt{n}\right) n' - n \left(\frac{n'}{2} - t \sqrt{n'}\right) \right) \\
    &= &
    \frac{t (n \sqrt{n'} + \sqrt{n} n')}{n(n-n')}  
    \leq
    \frac{2t \sqrt{\gamma} n^{3/2}} {(1-\gamma)n^2}
    \leq 
    4 t \sqrt{\frac{\gamma}{n} }\quad (\text{from }\gamma \leq \frac{1}{2})\ .
  \end{eqnarray*}
  Then, $\tau_{k,p}(\delta)$ in Lemma~\ref{lem:b-and-b} can be bounded by
  \begin{eqnarray*}
    \tau_{k,p}(\delta)
    &\leq&
    4t \sqrt{\frac{\gamma}{n}} \sqrt{\frac{k+2}{2p(1-p)}} 
    \leq
    4t \sqrt{\frac{3\gamma (k+2)}{n}} \quad (\text{from } \tfrac{1}{p(1-p)} < 6)\\
    &\leq&
    12t \sqrt{\gamma k/n} \leq 12t \gamma \quad (\text{from } 1 \leq k \leq \gamma n)\ .
  \end{eqnarray*}
  Note that, from the assumption, we have $\tau_{k,p}(\delta) \leq \frac{1}{2}$.
  By Lemmas~\ref{lem:h-and-b} and~\ref{lem:b-and-b},  
  we have 
  \begin{eqnarray*}
    \eqref{eq:first} 
    &\leq &
    \frac{k}{n} + \frac{\sqrt{e}}{2}\frac{\tau_{k,p}(\delta)}{(1-\tau_{k,p}(\delta))^2} + \frac{k}{n - n'}  \\
    &\leq&
    3 \gamma + 2\sqrt{e} \cdot 12 t \gamma \quad (\text{from } \tau_{k,p}(\delta) \leq \frac{1}{2})\\
    &\leq&
    c_{\ref{lem:dtv}}(1 + t) \gamma
  \end{eqnarray*}
  for some universal constant $c_{\ref{lem:dtv}}$.
\end{proof}

\section{Testing partial symmetry}

\subsection{Analysis of \textsc{Find-Asymmetric-Set}}
\label{appendix:fas}

In this section we prove there exists an algorithm \textsc{Find-Asymmetric-Set}, which satisfies Lemma~\ref{lem:properties-of-fas}.

Suppose that we have two inputs $x,y\in \{0,1\}^n$ with $x_{J} = y_{J}, |x| = |y|$ such that $f(x) \neq f(y)$.
Given such inputs, we know there exists some asymmetric variable outside of $J$. In order to efficiently find a set
from a partition $\calI$ which contains such a variable, we will use binary search over the sets. 
First, we construct a refinement $\calJ$ of $\calI$.
Every set of $\calI \setminus \{W\}$ is
partitioned further into parts so that each part has size at most $\lceil |W| / 4 \rceil$.
Let $t = |\calJ \setminus \{W\}|$ be the number of parts in $\calJ$ excluding the workspace.
Notice that the number of parts is at most $t \leq r + 4n / |W| = O(r)$.
Then, we construct a series of inputs $x^0 = x, x^1, \ldots, x^t=y$ by each step permuting only elements from some set $I \in \calJ \setminus\{W\}$ and the workspace $W$ (that is, applying a permutation from $\calS_{I \cup W}$). In each such step, we guarantee that $x^i_{I} = y_{I}$ for one more set $I \in \calJ \setminus \{W\}$, and therefore after (at most) $t$ steps we would reach $y$ (notice that we can choose the last step such that $x^t_W = y_W$ as the Hamming weight of all the inputs in the sequence is identical). 

Using this construction, we can now describe the algorithm \textsc{Find-Asymmetric-Set} as follows.
\begin{algorithm}
  \caption{\textsc{Find-Asymmetric-Set}$(f, \calI, J, W)$}
  \begin{algorithmic}
    \STATE Generate $x \in \{0,1\}^n$ and $\pi \in \calS_{\overline{J}}$ uniformly at random.
    \IF{$f(x) \neq f(\pi x)$}
     \STATE Define $x^0,\ldots,x^t$.
     \STATE Perform binary search on $x=x^0,\ldots,x^t=y$, and find $i$ such that $f(x^{i-1}) \neq f(x^i)$.
     \RETURN the only part $I \in \calI \setminus \{W\}$ such that $x^{i-1}_I \neq x^{i}_I$. 
    \ENDIF
    \RETURN $\emptyset$.
  \end{algorithmic}
\end{algorithm}

\begin{proof}[Proof of Lemma~\ref{lem:properties-of-fas}]
  Since we perform binary search over the sequence $x^0,\ldots,x^t$,
  the query complexity of the algorithm is indeed $O(\log t) = O(\log r)$. 
  Also, it is easy to verify that we only output an empty set or a part in $\calI \setminus \{W\}$ disjoint to $J$ (as $x_J = y_J$).

  Two random inputs $x$ and $y := \pi x$, for $\pi \in \calS_J$, satisfy $f(x) \neq f(y)$ with probability $\SymInf_f(\overline{J})$.
  Thus, it suffices to show that we can always define a sequence of $x^0,\ldots,x^{t}$,
  given that $|W| \geq \tfrac{n}{2r}$.
  In order to see this is always feasible, we consider the sequence after already defining $x^0,\ldots, x^i$,
  showing we can define $x^{i+1}$.
  
  Let $\calJ^+ = \{ I \in \calJ \mid | x^i_I | > | y_I | \}$ and $\calJ^- = \{ I \in \calJ \mid | x^i_I | < | y_I |\}$
  denote the sets which require increasing or decreasing the Hamming weight of $x_W$ respectively,
  when applying a permutation from $\calS_{I \cup W}$ to ensure $x^{i+1}_I = y_I$.
  Notice that we ignore sets $I$ for which $| x^i_I | = | y_I |$, as they do not impact the Hamming weight of $x^i_W$.
  If $| \calJ^+ | > 0$ and $| \calJ^- | > 0$, then since $\max( |x^i_W|, |W| - |x^i_W|) \geq \lceil |W| / 2 \rceil $ and the size of
  every set $I \in \calJ \setminus \{W\}$ is at most $\lceil |W|/4 \rceil $, there must exists a set we can use to define $x^{i+1}$.
  On the other hand, if $| \calJ^+ | = 0$ for example, then we can define $x^{i+1}$ using any set from $\calJ^-$ as
  $| x^i_W | - |y_W | = -\sum_{I \in \calJ \setminus \{W\}} | x^i_I | - | y_I |$
  (recall that $|x| = |x^i| = |y|$).
  
  It remains to show that when $W$ contains no asymmetric variables and we output a part $I \in \calI \setminus \{W\}$,
  $I$ contains an asymmetric variable.
  Suppose that the output $I$ is the part which was modified between $x^{i-1}$ and $x^{i}$.
  Then, since $f(x^{i-1}) \neq f(x^i), |x^{i-1}| = |x^i|$,
  and $x^{i-1}$ and $x^i$ differ only on $I \cup W$,
  an asymmetric variable exists in $I \cup W$ and we know it is not in $W$.
\end{proof}

\subsection{Proof of Lemma~\ref{lem:psf-main}}\label{appendix:psf-main-proof}
We first note that when the number of parts $r$ is bigger then $n$, we simply partition into the $n$ single-element sets
and the lemma trivially holds.
For $0 \le t \le 1$, let $\calF_t = \{J \subseteq [n] : \SymInf_f(\overline{J}) < t\eps,\ |J| \leq 5kn/r \}$ be the family of all sets which are not too big and whose complement has symmetric influence of at most $t\eps$. (Notice that with high probability, the union of any $k$ sets in the partition would have size smaller than $5kn/r$, and therefore we assume this is the case from this point on.) Our first observation is that for small enough values of $t$, $\calF_{t}$ is a $(k+1)$-intersecting family.
Indeed, for any sets $J, K \in \calF_{1/3}$,
$$
\SymInf_f(\overline{J \cap K}) = \SymInf_f(\overline{J} \cup \overline{K}) \leq \SymInf_f(\overline{J}) + \SymInf_f(\overline{K}) + c\sqrt{5k/r} < 2\eps / 3 + \eps / 9 < \eps\ .
$$
Since $f$ is $\eps$-far from $(n-k)$-symmetric, every set $S \subseteq [n]$ of size $|S| \leq k$ satisfies $\SymInf_f(\overline{S}) \geq \eps$.  So $|J \cap K| > k$.

We consider two cases separately: when $\calF_{1/3}$ contains a set of size less than $2k$; and when it does not. The first case is identical to the proof of Lemma~\ref{lem:junta-main} and hence we do not elaborate on it.

In the second case, which also resembles the proof of Lemma~\ref{lem:junta-main}, 
we claim that $\calF_{1/9}$ is a $2k$-intersecting family. If this was not the case, we could find sets $J, K \in \calF_{1/9}$ such that $|J \cap K| < 2k$ and $\SymInf_f(\overline{J \cap K}) \le \SymInf_f(\overline{J}) + \SymInf_f(\overline{K}) + \eps/9 < \eps/3$, contradicting our assumption.

Let $J \subseteq [n]$ be the union of $k$ parts in $\mathcal{I}$.  Since
$\mathcal{I}$ is a random partition, $J$ is a random subset obtained by
including each element of $[n]$ in $J$ independently with probability $p = k/r
< \frac1{2k+1}$. To bound the probability that $J$ contains some element from $\calF_{1/ 9}$, we define $\calF'_{ 1/ 9}$ to be all the sets that contain a member from $\calF_{1 /9}$. Since $\calF'_{1 /9}$ is also a $2k$-intersecting family, by Theorem~\ref{thm:DS}, for every such $J$ of size at most $5kn/r$, $\Pr[ \SymInf_f(\overline{J}) < \tfrac \eps 9] = \Pr[ J \in \calF_{1/9} ] \leq \mu_{k/r}(\calF'_{1/9}) \le \left(  k/ r \right)^{2k}$.
Applying the union bound over all possible choices for $k$ parts, $f$ will not satisfy the condition of the lemma with probability at most ${r \choose k} \left( \frac k r \right)^{2k} = O(k^{-k})$, which completes the proof of the lemma.

\section{Isomorphism testing and sampling partially symmetric functions}

\subsection{Properties of the sampling distribution}
\label{appendix:dist-diw}

We start this section with the following observation. 
When the number of parts $r$ reaches $n$ (or alternatively when $k = \Omega(\sqrt{n})$),
we consider the partition of $[n]$ into the $n$ single-element sets. Notice that when this is the partition,
then in fact $\calDIW$ is identical to $\calD^*_{k,n}$, making the following proposition trivial.
Therefore, in the proof we assume that $r < n$ and $k = O(\sqrt{n})$.

\begin{proof}[Proof of Proposition~\ref{prop:samples}]
We start with the first part of the proposition, showing $x$ is almost uniform.
Consider the following procedure to generate a random $\calI, W$ and $x$. We draw a random Hamming weight $w \sim \calB_{n,1/2}$ and define $x'$ to be the input consisting of $w$ ones followed by $n-w$ zeros. We choose a random partition $\calI'$ of $[n]$ into $r$ \emph{consecutive} parts $I_1,\ldots, I_r$ (i.e., $I_1 = \{1,2,  \ldots, |I_1| \}$ and $I_r = \{n - |I_r| + 1, \ldots, n\})$ according to the typical distribution of sizes in a random partition. Let the workspace $W'$ be the only part which contains the coordinate $w$ (or $I_1$ if $w=0$). We now apply a random permutation over $x'$, $\calI'$ and $W'$ to get $x$, $\calI$ and $W$.

It is clear the above procedure outputs a uniform $x$ as we applied a random permutation over $x'$, which had a binomial Hamming weight. The choice of $\calI$ was also done at random, considering the applied permutation over $\calI'$. The only difference is then in the choice of the workspace $W$, which can only be reflected in its size. However, when $r = o(\sqrt{n})$ we will choose the middle part as the workspace with probability $1-o(1)$, regardless of its size. In the remaining cases, since there are $n/r = \Omega(\sqrt{n})$ parts, the possible parts to be chosen as workspace are a small fraction among all parts, and therefore $W$ would be $o(1)$-close to being a random part.

\bigskip
Proving the second property of the proposition, we also consider two cases.
When $r = o(\sqrt{n})$, with probability $1-o(1)$, the workspace would have size $\omega(\sqrt{n})$ and
also $w = n/2 + O(\sqrt{n})$. In such a case, the $r-1$ parts (excluding the workspace) would be half zeros and half ones,
and the marginal distribution over the number of ones in $J$ would be $\calH_{r-1,(r-1)/2,k}$
(assuming the elements of $J$ are separated by $\calI$, which happens with probability $1-o(1)$).
By Lemma~\ref{lem:h-and-b}, the distance between this distribution and $\calB_{k,1/2}$ is bounded by $k/r < c/k$
for our choice of $0 < c < 1$.
Since there is no restriction on the ordering of the sets, this is also the distance from uniform over $\{0,1\}^k$ as required.

In the remaining case where $r = \Omega(\sqrt{n})$, we can use the same arguments and also apply Lemma~\ref{lem:b-and-b}
with the distributions $\calB_{k, 1/2}$ and $\calB_{k, 1/2+\delta}$ for $\delta = O(1/\sqrt{n})$, implying the distance
between these two distributions is at most $o(1)$. Combining this with the distance to $\calH_{r-1,(r-1)(1/2+\delta),k}$
we get again a total distance of $k/r + o(1) < c/k$ for our choice of $0 < c < 1$.
\end{proof}

\subsection{Analysis of \textsc{Partially-Symmetric-Isomorphism-Test}}
\label{appendix:iso-psf-proof}

The analysis of the algorithm is based on the fact that functions which passes the \textsc{Partially-Symmetric-Test}
satisfy some conditions, and in particularly are closed to being partially symmetric. We therefore start with the following lemma.

\begin{lemma}\label{lem:psf-close}
Let $g$ be a function $\eps$-close to being $(n-k)$-symmetric which passed the
\textsc{Partially-Symmetric-Test}$(g, k, \eps)$.
In addition, let $\calI, W$ and $J$ be the partition, workspace and identified parts used by the algorithm.
With probability at least $9/10$, there exists a function $h$ which satisfies the following properties.
\begin{itemize}
\setlength{\itemsep}{0pt}
\item $h$ is $4\eps$-close to $g$, and
\item $h$ is $(n-k)$-symmetric whose asymmetric variables are contained in $J$ and separated by $\calI$.
\end{itemize}
\end{lemma}

\begin{proof}
Let $g^*$ be the $(n-k)$-symmetric function closest to $g$ (which can be $f$ itself, up-to some isomorphism) and $R$ be the set of (at most) $k$ asymmetric variables of $g^*$. By Lemma~\ref{lem:distance-to-psf} and our assumption over $g$,
$$
\SymInf_g(\overline{R}) \leq 2 \cdot \dist(g, g^*) \leq 2\eps \ .
$$
Notice however that $R$ is not necessarily contained in $J$ and therefore $g^*$ is not a good enough candidate for $h$.
Let $U = R \cap J$ be the intersection of the asymmetric variables of $g^*$ and the sets identified by the algorithm.
In order to show that $g$ is also close to being $\overline{U}$-symmetric, we bound $\SymInf_g(\overline{U})$
using Lemma~\ref{lem:syminf-weak-subadditivity} with the sets $\overline{R}$ and $\overline{J}$.
Notice that since $|R| \leq k$ and $|J| \leq 2kn/r \leq \eps^2n/c'$ for our choice of $c'$, 
we can bound the error term (in the notation of Lemma~\ref{lem:syminf-weak-subadditivity}) by
$c\sqrt{\gamma} \leq c \sqrt{\eps^2/c'} \leq \eps$. We therefore have
\begin{eqnarray*}
  \SymInf_g(\overline{U}) \leq \SymInf_g(\overline{R}) + \SymInf_g(\overline{J}) + \eps 
  \leq
  2\eps + \eps + \eps = 4\eps
\end{eqnarray*}
where we know $\SymInf_g(\overline{J}) \leq \eps$ with probability at least $19/20$ as the algorithm did not reject.

By applying Lemma~\ref{lem:distance-to-psf} again, we know there exists a $\overline{U}$-symmetric function $h$, whose distance to $g$ is bounded by $\dist(g, h) \leq 4\eps$. Moreover, with probability at least $19/20$, all its asymmetric variables are completely separated by the partition $\calI$ (and they were all identified as part of $J$).
\end{proof}

Given Lemma~\ref{lem:psf-close}, we are now ready to analyze \textsc{Partially-Symmetric-Isomorphism-Test}.
\begin{proof}[Proof of Theorem~\ref{thm:psf-iso-test}]
Before analyzing the algorithm we just described, we consider the case where $k > n/10$.
Since Theorem~\ref{thm:psf-test} does not hold for such $k$'s, we apply the basic algorithm of
$O(n\log n / \eps)$ random queries, which is applicable testing isomorphism of any given function
(since there are $n!$ possible isomorphisms, the random queries will rule out all of them with good probability,
assuming we should reject).
Since $k = \Omega(n)$, the complexity of this algorithm fits the statement of our theorem.

We start by analyzing the query complexity of the algorithm. The step of \textsc{Partially-Symmetric-Test} performs
$O(\tfrac{k}{\eps}\log \tfrac{k}{\eps})$ queries, and therefore the majority of the queries are performed at the sampling stage,
resulting in $O(k\log k/\eps^2)$ queries as required. In order to prove the correctness of the algorithm, 
we consider the following cases.
\begin{itemize}
\setlength{\itemsep}{0pt}
\item $g$ is $\eps$-far from being isomorphic to $f$ and $\eps/1000$-far from being $(n-k)$-symmetric.
\item $g$ is $\eps$-far from being isomorphic to $f$ but $\eps/1000$-close to being $(n-k)$-symmetric.
\item $g$ is isomorphic to $f$.
\end{itemize}
In the first case, with probability at least $9/10$, \textsc{Partially-Symmetric-Test} will reject and so will we, as required. 
We assume from this point on that \textsc{Partially-Symmetric-Test} did not reject, as it will only reject $g$ which is
isomorphic to $f$ with probability at most $1/10$, and that we are not in the first case.
Notice that these cases match the conditions of Lemma~\ref{lem:psf-close},
and therefore from this point onward we assume there exists an $h$ satisfying the lemma's properties
(remembering we applied the algorithm with $\eps/1000$).

In order to bound the distance between $h$ and $g$ in our samples, we use Proposition~\ref{prop:samples}, indicating
$$
\Pr_{\calI, W \in \calI, x \sim \calDIW}[g(x) \neq h(x)] = \dist(g, h) + o(1/n) \ .
$$
By Markov's inequality, with probability at least $9/10$, the partition $\calI$ and the workspace $W$ satisfy
$$
\Pr_{x \sim \calDIW}[g(x) \neq h(x)] \leq 10 \cdot \dist(g, h) + o(1/n) \leq 10 \cdot 4\eps/1000 + o(1/n) < \eps/20 \ .
$$

By Proposition~\ref{prop:samples}, if we were to sample $h$ according to $\calDIW$, it should be $\eps/20$-close to sampling its core (assuming the partition size is large enough). Combined with the distance between $g$ and $h$ in our samples, we expect our samples to be $\eps/20 + \eps/20 = \eps/10$ close to sampling $h$'s core.

The last part of the proof is showing that there would be an almost consistent isomorphism of $f$ only when $g$ is isomorphic to $f$.
Notice however that we care only for isomorphisms which map the asymmetric variables of $f$ to the $k$ sets of $J$.
Therefore, the number of different isomorphisms we need to consider is $k!$.

Assume we are in the second case and $g$ is $\eps$-far from being isomorphic to $f$.
Let $f_{\pi}$ be some isomorphism of $f$. By our assumptions and Lemma~\ref{lem:psf-close},
$$
\dist(f_{\pi}, h) \geq \dist(f_{\pi}, g) - \dist(g,h) \geq \eps - \eps/250 \ .
$$
Each sample we perform would be inconsistent with $f_{\pi}$ with probability at least $\eps - \eps/250 - \eps/10 > 8\eps/9$.
By the Chernoff bounds and the union bound, if we would perform $q = O(k\log k / \eps^2)$ queries, we would rule
out all $k!$ possible isomorphisms with probability at least $9/10$ and reject the function as required.

On the other hand, if $g$ is isomorphic to $f$, then we know there exists with probability at least $9/10$
some isomorphism $f_{\pi}$ which maps the asymmetric variables of $f$ into the sets of $J$, such that
$$
\dist(f_{\pi}, h) \leq \dist(f_{\pi}, g) + \dist(g,h) \leq \eps/500 + \eps/250 \ .
$$
For this isomorphism, with high probability much more than $(1-\eps/2)$-fraction of the queries would be consistent
and we would therefore accept $g$ as we should.
\end{proof}

\subsection{Efficient sampler for partially symmetric functions}
\label{appendix:sampler}

We first provide the algorithm for efficiently generating a $\delta$-sampler for partially symmetric functions. The algorithm perform its preprocessing by calling \textsc{Partially-Symmetric-Test}. Given the output of the algorithm, we query the function once for each call to the sampler, according to $\calDIW$, and return the result.

\begin{algorithm}
  \caption{\textsc{Partially-Symmetric-Sampler}$(f, k, \delta, \eta)$}
  \begin{algorithmic}[1]
    \STATE Perform \textsc{Partially-Symmetric-Test}$(f, k, \eta \delta)$.
    \STATE Let $\calI$ and $W \in \calI$ be the partition and workspace used by the algorithm.
    \STATE Let $J$ be the union of $k$ parts in $\calI \setminus \{W\}$ that were identified by the algorithm.
    \STATE Return the following sampler:
    \STATE \ \ \ Choose a random $y \sim \calDIW$
    \STATE \ \ \ Let $x \in \{0,1\}^k$ be the value assigned to the parts in $J$
    \STATE \ \ \ Yield the triplet $(x, |y| - |x|, f(y))$
  \end{algorithmic}
\end{algorithm}

\begin{proof}[Proof of Theorem~\ref{thm:sampler}]
The algorithm for generating the sampler is described by \textsc{Partially-Symmetric-Sampler}, which performs $O(\tfrac {k}{\eta \delta} \log \tfrac{k}{\eta \delta})$ preprocessing queries to the function. What remains to be proved is that indeed with good probability, the algorithm returns a valid sampler.

Let $h$ be the function defined in the analysis of Theorem~\ref{thm:psf-iso-test}, which satisfies the conditions of
Lemma~\ref{lem:psf-close}. Recall that its asymmetric variables were separated by $\calI$ and appear in $J$.
Following this analysis and that of \textsc{Partially-Symmetric-Test},
one can see that with probability at least $1-\eta$ we would not reject $f$
when calling \textsc{Partially-Symmetric-Test}. Moreover, the samples would be $\delta/2$-close to sampling
the core of $h$, which is by itself $\delta/2$-close to $f$.
Therefore, overall our samples would be $\delta$-close to sampling the core of $f$.

The last part in completing the proof of the theorem is showing that we sample the core with distribution $\delta$-close
to $\calD^*_{k,n}$. By Proposition~\ref{prop:samples}, the total variation distance between sampling the core
according to $\calD^*_{k,n}$ and sampling it according to $\calDIW$ is at most $c/k$ for our choice of $0 < c < 1$,
which we can choose it to be at most $\delta$.
\end{proof}

Notice that if the function $f$ is not $(n-k)$-symmetric but still very close (say $(k/\eta \delta)^2$-close), applying the same
algorithm will provide a good sampler for an $(n-k)$-symmetric function $f'$ close to $f$. The main reason is that most likely, we will
not query any location of the function where it does not agree with $f'$.

\end{document}